\documentclass[review]{elsarticle}

\makeatletter
\def\ps@pprintTitle{%
 \let\@oddhead\@empty
 \let\@evenhead\@empty
 \def\@oddfoot{}%
 \let\@evenfoot\@oddfoot}
\makeatother
\pdfoutput=1

\usepackage{hyperref}
\hypersetup{colorlinks, breaklinks, citecolor=blue, linkcolor=red, urlcolor=blue}
\usepackage[left=1in,right=1in]{geometry}
\usepackage{amsmath}
\usepackage{graphicx}
\usepackage{amssymb}
\usepackage{amsthm}
\usepackage{gensymb}
\usepackage{filecontents,lipsum}
\usepackage{float}
\usepackage{color,colortbl}

\newtheorem{thms}{Theorem}[subsection]
\newtheorem{thmss}{Theorem}[subsubsection]

\newtheorem{assumps}{Assumption}[subsection]
\newtheorem{assumpss}{Assumption}[subsubsection]

\begin{document}

\newcommand*{\QEDA}{\hfill\ensuremath{\blacksquare}}
\newcommand*{\QEDB}{\hfill\ensuremath{\square}}
\definecolor{LightCyan}{rgb}{0.88,1,1}
\definecolor{Gray}{gray}{0.85}

\begin{frontmatter}

\title{Control of a nonlinear continuous stirred tank reactor via event triggered sliding modes}




%
%

\author[asinha]{Abhinav Sinha\corref{cor1}}
\ead{sinha.abhinav.pg2016@mechatronics.iiests.ac.in}
\author[rkmJ,rkmI]{Rajiv Kumar Mishra}
\ead{mishra@sc.dis.titech.ac.jp}

\cortext[cor1]{Corresponding author}
\address[asinha]{Indian Institute of Engineering Science and Technology, Shibpur (Howrah), India}
\address[rkmJ]{Department of Computer Science, School of Computing, Tokyo Institute of Technology, Japan}
\address[rkmI]{School of Electronics Engineering, Kalinga Institute of Industrial Technology, Bhubaneswar, India}

\begin{abstract}
Continuous Stirred Tank Reactors (CSTR) are the most important and central equipment in many chemical and biochemical industries, that exhibit second order complex nonlinear dynamics. The nonlinear dynamics of CSTR poses many design and control challenges. The proposed controller guarantees a stable closed loop behavior over multiple operating points even in the presence of perturbations and parametric anomalies. An event driven sliding mode control is presented in this work to regulate the temperature and concentration states very close to the equilibrium points of a CSTR. The control is executed only when a predefined condition gets violated and hence, the controller is relaxed when the system is operating under tolerable limits of closed loop performance. A novel dynamic event triggering rule is presented to maintain desired performance with minimum computational cost. The inter event execution time is shown to be lower bounded by a finite positive quantity to exclude Zeno behavior. Sliding mode control (SMC) combined with event triggering scheme retains the inherent robustness of traditional SMC and aids in reducing computational load on the controller involved. Simulation results validate the efficiency of the proposed controller.
\end{abstract}

\begin{keyword}
CSTR, event-based sliding mode control, Riemann sampling, Lebesgue sampling, event conditions, triggering rule, inter event time.
\end{keyword}

\end{frontmatter}


\section{Introduction}
A Continuous Stirred Tank Reactor (CSTR) exhibits complex nonlinear dynamics and is a benchmark equipment in many process industries \cite{Sinha2016,sinhaCSTRacods} that require continuous addition and withdrawal of reactants and products. To maximize economy and to achieve optimal productivity in chemical plants, these reactors are maintained at very high conversion rates. A CSTR may be assumed to be somewhat opposite of an idealized well-stirred batch and tubular plug-flow reactors. The set of operating points should exhibit a stable steady state behavior under the influence of disturbances as well. Linear controllers designed for such process fail to deliver optimal performance outside the linear operating range. The PID controller is the most commonly used controller in industries due to its easy design and tuning properties. Feedback linearization has been also extensively used to control CSTR \cite{zhai} wherein the controller fails to deliver under varying transient behavior of the plant model due to the non-adaptive nature of the controller. A linear controller using Taylor's linearization has been designed assuming bounded uncertainty in \cite{AIC:AIC690340708} and in \cite{ref18326137}, which are again based on operation of CSTR in a limited regime. It has been discussed in \cite{ALVAREZRAMIREZ19941743} that use of local linearization cannot ensure global stability. Input-output feedback linearization method proposed in \cite{acs} also failed because it requires continuous measurement of states, which is quite expensive and impractical in practical scenario. In \cite{acs}, a method based on state coordinate transformation has been studied for linear input state behavior. Observer based designs for state measurements have been discussed in \cite{doi:10.1021/ie990186e,CHEN2004501,doi:10.1021/ie0608713,GRAICHEN2009473,DICICCIO,HOANG,ANTONELLI} to name a few. While the method in aforementioned studies ensured asymptotic stability, the effects of disturbance was overlooked and the design failed to deliver desired response under varying process conditions. A high gain controller used in \cite{ichem} exhibited quicker response but resulted in unwanted control effort saturation. While performing linearization of a plant model, there remains a part of transformed system which is non linearizable \cite{ray2} and has zero dynamics which cannot be ignored \cite{acs}. In spite of being a non-model based (model-free) control, techniques such as adaptive and fuzzy control do not yield optimal performance under widely varying and fast process dynamics.\\

In this work, we propose a controller based on paradigms of event based sliding mode control. Sliding Mode Control (SMC) \cite{yan2017,sks1998,pbbce066e,utkin,intro,smc,zak,slot} is a control scheme which guarantees finite time convergence and provides robust operation over the entire regime with complete rejection to matched perturbations that may creep in the system from input channel. The advantage of using this control is that we can tailor the system dynamical behavior by a particular choice of sliding function. SMC, used in conjunction with event triggered control, retains its robustness, as well as event triggering approach aids in saving energy expenditure. When measured variables of a system do not deviate frequently, event based control offers numerous advantages over time triggered control. The control is executed only when needed, thus computational complexity is also reduced.

%
%
%
%
%

\section{Plant Dynamic Model}
Chemical reactions in a reactor can be characterized as endothermic or exothermic. In order to maintain the temperature of the reactor at a desired reference, a finite amount of energy is required to be added to or removed from the reactor. Usually, a CSTR operates at steady state with contents well mixed, so modeling does not involve significant variations in concentration, temperature or reaction rate throughout the vessel. Since internal states of the system under consideration, i.e, temperature and concentration are identical everywhere within the reaction vessel, they are the same at the exit as they are anywhere else in the tank. Consequently, the temperature and concentration at the exit are modeled as being the same as those inside the reactor. In situations where mixing is highly nonideal, the well mixed model fails and nonideal CSTR model must be formulated.\\

In our study, we have adopted a nonlinear model of a CSTR under the assumption that contents are well mixed and hence, temperature and concentration are identical everywhere within the reaction vessel. This assumption simplifies the analysis of the model under consideration. However, if these assumptions are violated, then one has to resort to non-ideal mathematical model of CSTR. In a non-ideal CSTR, mixing is not uniform throughout the vessel and as a consequence, there occurs bypassing and dead zones (stagnant regions). The fluid does not pass through the stagnant region, resulting in lesser volume of the CSTR than that in the case of an ideal CSTR with perfect mixing. Due to this, the fluid passes through the CSTR rapidly and the transients in the concentration die very quickly. Moreover, the flow through the reactor also becomes less than the total volumetric flow rate in the CSTR due to bypassing. Once again, the transients in the concentration die out quickly.\\

In this work, we are concerned with a dynamic description of the reactor in which mixing is adequate \cite{Sinha2016}. Thus, an ideal CSTR model as provided in \cite{ray} has been adopted in our study. Presence of exponential terms in the modeling equations make the description a nonlinear one. A complex chemical reaction occurs in CSTR, \emph{e.g.} conversion of a hazardous chemical waste (reactant) into an acceptable and tolerable chemical (product). Under the assumption of complete mixing, the reactor gets cooled in a continuous manner. The volume of the chemical product \textbf{B} is equal to the volume of the input reactant \textbf{A}. The reactor is assumed to be non isothermal and exhibiting an irreversible exothermic first order chemical reaction \textbf{A} $\rightarrow$ \textbf{B}.\\
The dynamic model is then given as
\begin{align}\label{cstrdynamicmodel}
\frac{dC_A}{dt'} &= \frac{F}{V} (C_{A_f} - C_A) - r \nonumber \\
\frac{dT}{dt'} &= \frac{F}{V} (T_f - T) + \frac{(-\Delta H)}{\rho C_p} r - \frac{hA}{V\rho C_p} (T-T_c)
\end{align}
\begin{equation}\label{eq:rate}
r = k_0 exp({-\frac{E}{RT}}) C_A
\end{equation}
\begin{equation}\label{eq:heatTransferTerm}
hA = \frac{aF_c^{b+1}}{F_c + \frac{aF_c^b}{2\rho_c C_{p_c}}}
\end{equation}
$a$,$b$ are CSTR model parameters and $hA$ is the heat transfer term in (\ref{eq:heatTransferTerm}). Model parameters of significance are given in table \ref{tblcheck:parametermeaning} and a schematic diagram of CSTR is shown in figure \ref{fig:cstrschematics}.
\begin{figure}[H]
\centering
\includegraphics{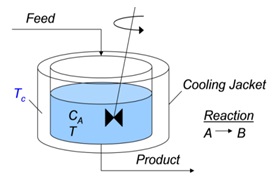}
\caption{Schematic of a CSTR}
\label{fig:cstrschematics}
\end{figure}
\begin{table}[H]
\centering
\setlength{\doublerulesep}{2\arrayrulewidth}
\begin{tabular}{|| l | c | r ||}
  \hline \hline
\textbf{Meaning}&\textbf{Symbol}&\textbf{Unit} \\ \hline

  $1^{st}$ order reaction rate constant & $k_0$ & $min^{-1}$\\ \hline

  inlet concentration of \textbf{A} & $C_{A_f}$ & $kmol/m^3$ \\ \hline

steady state flow rate of \textbf{A} & $F$ & $m^3/min$ \\ \hline

density of the reagent \textbf{A} & $\rho$ & $g/m^3$ \\ \hline

specific heat capacity of \textbf{A} & $C_p$ & $cal/\degree C g $ \\ \hline

heat of reaction & $\Delta H$ & $cal/kmol$ \\ \hline

density of coolant& $\rho_c$ & $g/m^3$ \\ \hline

specific heat capacity of coolant& $C_{p_c}$ & $cal/\degree C g$ \\ \hline

volume of the CSTR& $V$ & $m^3$ \\ \hline

coolant flow rate& $F_c$ & $m^3/min$ \\ \hline

reactor temperature& $T$ & $K$ \\ \hline

reactor concentration of \textbf{A}& $C_A$ & $kmol/m^3$ \\ \hline

activation energy& $E$ & $J/mol$ \\ \hline

universal ideal gas constant& $R$ & $J/mol K$ \\
  \hline \hline
\end{tabular}
\caption{CSTR model parameters and their meanings}
\label{tblcheck:parametermeaning}
\end{table}

A computationally more convenient form of the above modeling equations are presented in state space formulation below. For original convention and nomenclature, the reader is suggested to refer \cite{ray}.
\begin{align}\label{eq:statemodelcstr}
\dot{x_1}&=-x_1 + D_a(1-x_1) \exp\Big(\frac{x_2}{1+x_2 / \gamma}\Big) - d_2 \nonumber\\
\dot{x_2}&=-x_2 + BD_a(1-x_1) \exp\Big(\frac{x_2}{1+x_2 / \gamma}\Big) - \beta (x_2 -  x_{2_{c_0}})+ \beta u_T + d_1
\end{align}
This formulation \ref{eq:statemodelcstr} utilizes dimensionless modeling of CSTR for which the parameters are tabulated in table \ref{tbl:dimensionlessparamters}. Parameters $d_1$ and $d_2$ are bounded and measurable disturbances. The nominal coolant temperature is $x_{2_{c_0}}$. It is worthy to note that there are two ways to manipulate the observed states (outputs)-- coolant temperature and input feed flow. In this study, we have used the coolant temperature as our control input to regulate the temperature of the CSTR. This is because, in industrial environments, temperature becomes more critical to be controlled in order to avoid any secondary reaction in the reactor. It should be noted that while designing a controller to regulate temperature, the other state, i.e., composition should not be allowed to enter the region of instability. Furthermore, design of a second controller to regulate composition is also discussed in this work.
\begin{table}[H]\label{tbl:dimensionlessparamters}
\centering
\setlength{\doublerulesep}{2\arrayrulewidth}
\begin{tabular}{ ||l|c|| }
 \hline
  \hline
  Description & Parameter \\ \hline \hline
 ratio of activation energy to average kinetic energy & $\gamma = E/RT_{f_0}$ \\
  adiabatic temperature rise & $B = \frac{(-\Delta H)C_A{}_{_{f_{}{_0}}} \gamma}{\rho C_p T_f{}_{_0}}$ \\

  Damkohler number & $D_a = k_0 exp(-\gamma )V/ F_0$ \\

  heat transfer coefficient & $\beta = hA/\rho C_p F_0$ \\

  dimensionless time & $t=t'(F_0/V)$ \\

  dimensionless composition & $x_1 = (C_A{}_{_{f_{}{_0}}} - C_A)/C_A{}_{_{f_{}{_0}}}$ \\

  dimensionless temperature & $x_2 = \gamma(T-T_f{}_{_0})/T_f{}_{_0}$ \\

  dimensionless control input & $u_T = \gamma(T-T_c{}_{_0})/T_f{}_{_0}$ \\

 dimensionless control input & $u_F = (F-F_0)/F_0$ \\

  feed temperature disturbance & $d_1 = \gamma (T_f - T_f{}_{_0})/T_f{}_{_0}$ \\

  feed composition disturbance & $d_2 = (C_A{}_{_f}- C_A{}_{_{f_{}{_0}}})/C_A{}_{_{f_{}{_0}}}$ \\
  \hline \hline
\end{tabular}
\caption{Dimensionless parameters in CSTR modeling}
\end{table}
For ease of controller synthesis, let us formulate a functional description of the state equations given by (\ref{eq:statemodelcstr}).
\begin{align}\label{eq:statemodelcstrfunctional}
 \dot{x_1}&= f_1(x_1,x_2) - d_2 \nonumber \\
 \dot{x_2}&= f_2(x_1,x_2) + \beta u_T + d_1
\end{align}

\section{Controller Synthesis}
The controller in this study is synthesized without any linearization of the dynamics. An event based sliding mode control has been used to implement the controller.
\subsection{Control Objective}
Primary objective of the control scheme is to maintain the states to a desired reference value with minimum computational expense. For computational purposes, it is desirable to define the following error candidates. The deviation from the desired temperature is given by
\begin{equation}\label{eq:temperatureError}
e_2(t)= x_2(t) - x_{2_{ref}} (t)
\end{equation}
Similarly, error variable for the state representing composition is given by
\begin{equation}\label{eq:compositionError}
e_1(t)= x_1(t) - x_{1_{ref}} (t)
\end{equation}
The control effort must be designed robust enough to achieve accurate desired reference tracking, reject disturbances and deliver acceptable results quickly. Stated alternatively, the error variables are required to vanish or at least settle in close vicinity of zero after a transient of acceptable duration.

\subsection{Traditional sliding mode controller}
Sliding Mode Control (SMC) \cite{utkin,utkin2} is long known for its inherent robustness. The switching nature of the control is used to nullify exogenous bounded disturbances and matched uncertainties. The switching happens about a manifold in state space known as sliding manifold. The control forces the state trajectories monotonically towards the sliding manifold and this phase is regarded as reaching phase. When state trajectories reach the manifold, they remain there for all future time, thereby ensuring that the system dynamics remains independent of bounded disturbances and matched uncertainties. This phase is regarded as sliding phase. Thus, the controller has a reaching phase (trajectories in phase space emanate and move towards the sliding manifold) and a sliding phase (trajectories in the phase space that reach the sliding manifold try to remain there).

\subsubsection{Reaching Phase}
Let the manifold discussed above be described mathematically as $\sigma(x)$. In order to drive state trajectories onto this manifold, a proper discontinuous control effort $u(t,x)$ needs to be synthesized for which the following inequality is respected.
\begin{equation}
\sigma^T(x)\dot{\sigma}(x) \leq -\eta \|\sigma(x)\|
\end{equation}
with $\eta$ being positive and is called the reachability constant.\\
$\because$
\begin{equation}
\dot{\sigma}(x) = \frac{\partial \sigma}{\partial x} \dot{x} = \frac{\partial \sigma}{\partial x} f(t,x,u)
\end{equation}
$\therefore$
\begin{equation}
\sigma^T(x) \frac{\partial \sigma}{\partial x} f(t,x,u) \leq -\eta \|\sigma(x)\|
\end{equation}
It is clear that the control $u$ can be synthesized from the above equation.

\subsubsection{Sliding Phase}
The motion of state trajectories confined on the switching manifold is known as \emph{sliding}. A sliding mode is said to exist in close vicinity of the manifold if the state velocity vectors are directed towards the manifold in its neighbourhood \cite{zak,utkin}. Under this circumstance, the manifold is called attractive \cite{zak}, i.e., trajectories starting on it remain there for all future time and trajectories starting outside it tend to it in an asymptotic manner.
\begin{align}
\because \hspace{3mm}\dot{\sigma}(x) = \frac{\partial \sigma}{\partial x} f(t,x,u) \nonumber \\
\text{Hence, in sliding motion} \nonumber \\
\frac{\partial \sigma}{\partial x} f(t,x,u) = 0
\end{align}
Then $u = u_{eq}$ (say) be a solution and is generally referred to as the equivalent control. This $u_{eq}$ is not the actual control applied to the system but can be thought of as a control that must be applied on an average to maintain sliding motion. It is mainly used for the analysis of sliding motion \cite{yan2017}. \\

By the theory of sliding modes, let us formulate the sliding manifold as
\begin{equation}\label{eq:slidingManifold}
\sigma (t) = \lambda_1 x_1(t) + \lambda_2 x_2(t)
\end{equation}
where $\lambda_1$ and $\lambda_2$ are the coefficient weights which can be tuned as per performance needs.In design, it is not the actual weights that matter, rather relative weights are of significance. The sliding manifold can alternatively be written in error dynamical form as
\begin{equation}\label{eq:slidingManifoldError}
\sigma (t) = \lambda_1 e_1(t) + \lambda_2 e_2(t)
\end{equation}
During sliding, $\dot{\sigma} (t) = 0$ and the corrective term used to force the trajectories onto the sliding surface is chosen as $\mu sign(\sigma (t))$, where $\mu$ is the adjustable gain.
We now proceed to design the control law for the discussed case.
\begin{thms}\label{thm:controlLaw}
Given plant dynamics (\ref{eq:statemodelcstrfunctional}), errors (\ref{eq:temperatureError},\ref{eq:compositionError}) and the sliding manifold (\ref{eq:slidingManifoldError}), the control law due to traditional sliding mode controller is given by
\begin{equation}\label{eq:controlLaw}
u_T (t) = -\lambda_2^{-1} \beta^{-1}(\lambda^{T} f(x(t)) + \mu sign(\sigma(t)))
\end{equation}
where $\lambda^{T} = [\lambda_1 \hspace{2mm} \lambda_2]$ and the function $f(x)$ is given as
\begin{equation}\label{f(x)}
f(x) =
\begin{bmatrix}
f_1(x_1,x_2) - d_2 - \dot{x}_{1_{ref}}\\
f_2(x_1,x_2) +d_1 -\dot{x}_{1_{ref}}
\end{bmatrix}
\end{equation}
\end{thms}
\begin{proof}
Before proceeding towards proof, we make an assumption on the nature of the function $f(x)$.
\begin{assumps}
The function $f(x)$ satisfies Lipschitz conditions. Hence, we can write $\| f(x) - f(y)\| \leq \bar{L}\|x - y\|$ for some $x$ and $y$ in the domain $\mathbb{D_L} \subset \mathbb{R}^n$ with $\bar{L}$ as the Lipschitz constant.
\end{assumps}
We now proceed to a formal proof. From (\ref{eq:slidingManifoldError}), we have
 \begin{align}\label{eq:sigmadot}
\sigma (t) &= \lambda_1 e_1(t) + \lambda_2 e_2(t) \nonumber \\
\Rightarrow \dot{\sigma} (t) &= \lambda_1 \dot{e}_1(t) + \lambda_2 \dot{e}_2 (t)\nonumber \\
\Rightarrow \dot{\sigma} (t) &= \lambda_1 (\dot{x}_1 (t)- \dot{x}_{1_{ref}})  + \lambda_2 (\dot{x}_2(t) - \dot{x}_{2_{ref}}) \nonumber \\
\Rightarrow \dot{\sigma} (t) &=\lambda_1 (f_1(x_1,x_2) - d_2 - \dot{x}_{1_{ref}}) + \lambda_2 (f_2(x_1,x_2) + \beta u_T + d_1 - \dot{x}_{2_{ref}}) \nonumber \\
\Rightarrow \dot{\sigma} (t) &=\lambda^{T} f(x(t)) + \lambda_2 \beta u_T(t) \\
\therefore u_T (t) &= -\lambda_2^{-1} \beta^{-1}(\lambda^{T} f(x(t)) + \mu sign(\sigma(t))) \nonumber
\end{align}
where $\lambda^{T} = [\lambda_1 \hspace{2mm} \lambda_2]$ and $f(x) =
\begin{bmatrix}
f_1(x_1,x_2) - d_2 - \dot{x}_{1_{ref}}\\
f_2(x_1,x_2) +d_1 -\dot{x}_{1_{ref}}
\end{bmatrix}$\\
This completes the theorem along its proof.
\end{proof}

\subsection{Event based control}
Recently, there has been a tremendous growth of interest in the area of event based systems due to requirements of reduced computational cost. The challenge, however, in this type of control is to maintain performance, stability, optimality, etc. in the presence of uncertainties and reduced computation/communication. A modern control system consists of a computer and the signal under consideration is sampled periodically to cater the needs of a classic sampled data control system. Under such scheme, the interval between two successive clock pulses is predetermined and fixed. The sampling takes place along the \emph{horizontal} axis, also known as \emph{Riemann sampling}. An alternate, more natural and efficient way is to sample along the \emph{vertical} axis, also known as \emph{Lebesgue sampling} \cite{1184824}. In the latter case, the sampling is not periodic rather it depends on the value of previous sample or certain \emph{conditions} that need to be violated to bring forth the next clock pulse. These \emph{conditions} are some noticeable changes (\emph{events} or \emph{event conditions}) on which the next sampling instant depends.\\

This type of control seems to be a reasonable choice in applications where signal of interest slowly varies. In chemical process industries that contain many production units, primary units are separated by buffer units. Each change in the unit can cause upset and hence it is desirable to keep the change in process variables less frequent. Event based control comes handy in such applications. No action is taken unless there is a huge upset. It is also advantageous to use event based control when the steady state value of a process variable needs to be fixed irrespective of the manner in which the states evolve. For early contributions on event based control, readers are requested to refer \cite{marconi,behera,lingshi,mazo,tabuda,Aström2008,anta,chopra,lemmon2010} and references therein.

\subsubsection{Event based sliding mode control}
Since, next sample instant is dependent on the previous sampling information, the control (\ref{eq:controlLaw}) is held constant between successive events or sampling instants. The control is not updated periodically and is held at the previous value in the interval $[t_k,t_{k+1})$ . This, however, introduces a discretization error between the states of the system.
\begin{equation}\label{eq:discretizationError}
\epsilon(t) = x(t)- x(t_k)
\end{equation}
such that at $t = t_k$, $\epsilon(t)$ vanishes. The term $t_k$ is the triggering instant at $k^{th}$ sampling instant. The control gets updated at $t_k$ instants only. The sampling is not periodic and hence $t_{k+1} - t_k \neq constant$.\\
Hence, the control signal from (\ref{eq:controlLaw}) modifies to yield the event triggered sliding mode control law
\begin{equation}\label{eq:controlLawEvent}
u_T (t) = -\lambda_2^{-1} \beta^{-1}(\lambda^{T} f(x(t_k)) + \mu sign(\sigma(t_k)))
\end{equation}
\begin{thmss}
Consider the system described by (\ref{eq:statemodelcstrfunctional}), error candidates (\ref{eq:temperatureError},\ref{eq:compositionError}) and (\ref{eq:discretizationError}), sliding manifold (\ref{eq:slidingManifoldError}) and control law of (\ref{eq:controlLawEvent}). Then, the event triggered control law (\ref{eq:controlLawEvent}) makes the system stable in the sense of Lyapunov and sliding mode is said to exist in the vicinity of the manifold (\ref{eq:slidingManifoldError}). The manifold is an attractor if reachability to the surface is ascertained for some reachability constant $\eta>0$.
\end{thmss}
\begin{proof}
 Let us consider a Lyapunov candidate $V$ such that
\begin{equation}\label{eq:LyapunovCandidate}
V = \frac{1}{2}\sigma^T(t) \sigma(t)
\end{equation}
Time derivative of the candidate given in (\ref{eq:LyapunovCandidate}) for $t \in [t_k, t_{k+1})$ along the state trajectories yield
\begin{equation}
\dot{V} = \sigma(t) \dot{\sigma}(t)
\end{equation}
It can be written from (\ref{eq:sigmadot}),
\begin{equation}
 \dot{V} = \sigma(t) (\lambda^{T} f(x(t)) + \lambda_2 \beta u_T(t))
\end{equation}
Thus $\forall t \in [t_k, t_{k+1})$, it can be written as
\begin{align}
\dot{V} &= \sigma(t) (\lambda^{T} f(x(t)) - \lambda^{T} f(x(t_k)) - \mu sign(\sigma(t_k))) \nonumber \\
\dot{V} &\leq -\sigma(t) \mu sign(\sigma(t_k)) + \| \sigma(t) \| \| \lambda^T \| \| f(x(t)) - f(x(t_k)) \| \nonumber \\
\dot{V} &\leq -\sigma(t) \mu sign(\sigma(t_k)) + \| \sigma(t) \| \| \lambda^T \| \bar{L} \| x(t)- x(t_k) \| \nonumber \\
\dot{V} &\leq -\sigma(t) \mu sign(\sigma(t_k)) + \| \sigma(t) \| \| \lambda^T \| \bar{L}  \| \epsilon(t) \| 
\end{align}
As long as $\sigma(t) > 0$ or $\sigma(t) < 0$, the condition $sign(\sigma(t)) = sign(\sigma(t_k))$ is strictly met $\forall t \in [t_k, t_{k+1})$. Hence, when trajectories are just outside the sliding surface,
\begin{align}
\dot{V} &\leq -\| \sigma(t) \| \mu + \| \sigma(t) \| \| \lambda^T \|  \bar{L} \| \epsilon(t) \|  \nonumber \\
\Rightarrow \dot{V} &\leq -\| \sigma(t) \| (\mu +  \| \lambda^T \|  \bar{L} \| \epsilon(t) \|)  \nonumber \\
\Rightarrow \dot{V} &\leq - \eta \| \sigma(t) \|
\end{align}
with $\eta > 0$. This completes the proof of reachability. \QEDB\\
For stability, it is required to be shown that $\dot{V}<0$.\\
At $t=t_k$, $\| \epsilon(t)\| \rightarrow 0$ and the control signal is updated. Thus,
\begin{equation}
\dot{V} \leq -\| \sigma(t) \| (\mu +  \| \lambda^T \| \bar{L} \| \epsilon(t) \|)
\end{equation}
$\because \| \epsilon(t)\| \rightarrow 0 \Rightarrow \dot{V} < 0$ \\
This completes the proof of stability.
\end{proof}
For time instants between $[t_k,t_{k+1})$ the states show a tendency to deviate from the sliding manifold but remain bounded within a band near the manifold. The triggering instant $t_k$ is completely characterized by a triggering rule. Next sampling instant is by virtue of this criterion. As long as this criterion is respected, next clock pulse is not called upon and the control signal is maintained constant at the previous value. The triggering rule used in this work is given by
\begin{equation}\label{eq:triggeringRule}
\delta = \| \zeta e_i + \xi \dot{e}_i^2 \| - \psi (m_1 + m_2 \exp(-\varsigma t))
\end{equation}
with $i = 1,2$ for respective error in states, $\zeta > 0$, $\xi > 0$, $\psi \in (0,1)$, $m_1 \geq 0$, $m_2 \geq 0$, $m_1 + m_2 > 0$ and $\varsigma \in (0,1)$. The term $(m_1 + m_2 \exp(-\varsigma t))$ ensures a finite lower bound on inter event execution time and avoids accumulation of samples at same instant, known as \emph{Zeno behavior} in literature. The triggering scheme has been developed to schedule controller updates based on triggering of an event so as to achieve better resource efficiency. The temperature control loops usually have smooth measurements and large time constants, so slope of the error can be used to predict future error(s). Based on the current slope of the error, the controller predicts the error in future and adds additional control action to the controller output to make the loop more stable and to reduce the maximum deviation of process variable from set-point. The following relation completely determines the triggering instants in an iterative manner
\begin{equation}\label{eq:iterativeTriggeringRule}
t_{k+1} = \text{inf} \{t \in [t_k, \infty) \hspace{2mm}: \delta \geq 0\}
\end{equation}

The iterative relation (\ref{eq:iterativeTriggeringRule}) simply implies that the next sample at time $t=t_{k+1}$ appears when the parameter $\delta$ (from (\ref{eq:triggeringRule})) is positive, i.e., $\| \zeta e_i + \xi \dot{e}_i^2 \|$ exceeds the tolerable performance band $\psi (m_1 + m_2 \exp(-\varsigma t))$. The threshold parameter considered in this event-driven triggering scheme is not hardbound but is time varying. The accuracy adjustment parameter $\psi (m_1 + m_2 \exp(-\varsigma t))$ is chosen such that the tolerable limit becomes thinner as steady state is being reached. The inter event time is given by
\begin{equation}\label{eq:intereventTime}
T_k = t_{k+1} - t_k
\end{equation}
\begin{assumpss}
A finite but not necessarily constant delay $\Delta$ might occur during sampling and is unavoidable due to hardware characteristics. In such cases the control is maintained constant $\forall t_i \in [t_i^k + \Delta, t_i^{k+1} + \Delta)$. It has been assumed that $\Delta$ is negligible and has been neglected innocuously. Hence for our case, control is constant in the interval $[t_i^k, t_i^{k+1})$.
\end{assumpss}
\begin{thmss}
Consider the system described by (\ref{eq:statemodelcstrfunctional}), the control signal given in (\ref{eq:controlLawEvent}) and the discretization error (\ref{eq:discretizationError}). The sequence of triggering instants $\{t_k\}_{k=0}^\infty$ respects the triggering rule given in (\ref{eq:triggeringRule}). Consequently, \emph{Zeno} phenomenon is not exhibited and the inter event execution time $T_k$ is bounded below by a finite positive quantity such that
\begin{equation}\label{eq:eventLowerBound}
T_k \geq \frac{1}{\bar{L}}\hspace{2mm} ln \hspace{2mm} \Big( \frac{\bar{L}\|\epsilon\|_\infty}{\bar{L}(1 + \|\bar{B}\lambda_2^{-1} \beta^{-1}\lambda^{T}\|) \|x(t_k)\| + \|\bar{B}\| \mu} +1 \Big)
\end{equation}
where $\|\epsilon\|_\infty$ is the maximum discretization error.
\end{thmss}
\begin{proof}
Without loss of generality, the system described in (\ref{eq:statemodelcstrfunctional}) is recalled here as
\begin{equation}\label{eq:generalPlantModel}
\dot{x}(t) = f(x) + \bar{B}u_T(t)
\end{equation}
where $f(x)$ is same as (\ref{f(x)}) and $\bar{B} = [0 \hspace{3mm}\beta]^T$. Between $k^{th}$ and $(k+1)^{th}$ sampling instant in the execution of control, the discretization error (\ref{eq:discretizationError}) is non zero. $T_k$ is the time it takes the discretization error to rise from $0$ to some finite value. Thus,
\begin{equation}
 \frac{d}{dt}\|\epsilon(t)\| \leq \|\frac{d}{dt}\bar{\epsilon}(t)\| \leq \|\frac{d}{dt}x(t)\| 
\end{equation}
\begin{equation}
\Rightarrow \|\frac{d}{dt}\epsilon(t)\| \leq \| f(x(t)) + \bar{B} u_T(t) \|
\end{equation}
Substituting the control input (\ref{eq:controlLawEvent}) in the above inequality, we get
\begin{align}\label{eq:differentialInequality}
\|\frac{d}{dt}\epsilon(t)\| &\leq \| f(x(t)) - \bar{B}\lambda_2^{-1} \beta^{-1}\lambda^{T} f(x(t_k)) -\bar{B} \mu sign(\sigma(t_k)) \| \nonumber \\
 &\leq \bar{L}\|x(t)\| + \|\bar{B}\lambda_2^{-1} \beta^{-1}\lambda^{T}\| \bar{L} \|x(t_k)\| + \|\bar{B}\| \mu \nonumber \\
& \leq \bar{L}(\|\epsilon(t)\| + \|x(t_k)\|) +  \|\bar{B}\lambda_2^{-1} \beta^{-1}\lambda^{T}\| \bar{L} \|x(t_k)\| + \|\bar{B}\| \mu \nonumber \\
&\leq \bar{L}\|\epsilon(t)\| + \bar{L}(1 + \|\bar{B}\lambda_2^{-1} \beta^{-1}\lambda^{T}\|) \|x(t_k)\| + \|\bar{B}\| \mu
\end{align}
The solution to the differential inequality (\ref{eq:differentialInequality}) $\forall t \in [t_k, t_{k+1})$ can be understood by using Comparison Lemma \cite{khalil} with initial condition $\|\epsilon(t)\|=0$ and is given as
\begin{equation}
\|\epsilon(t)\| \leq \frac{\bar{L}(1 + \|\bar{B}\lambda_2^{-1} \beta^{-1}\lambda^{T}\|) \|x(t_k)\| + \|\bar{B}\| \mu}{\bar{L}} \Big( exp\{\bar{L}(t-t_k)\} - 1\Big)
\end{equation}
Comparison Lemma \cite{khalil}, \cite{alexander} is particularly useful when information on bounds on the solution is of greater significance than the solution itself.
For triggering time instant $t_{k+1}$,
\begin{equation}
\|\epsilon\|_\infty = \|\epsilon(t_{k+1})\| \leq \frac{\bar{L}(1 + \|\bar{B}\lambda_2^{-1} \beta^{-1}\lambda^{T}\|) \|x(t_k)\| + \|\bar{B}\| \mu}{\bar{L}} \Big( \exp\{\bar{L}T_k\} - 1\Big)
\end{equation}
\begin{equation}\label{eq:Tk}
\therefore \hspace{2mm} T_k \geq \frac{1}{\bar{L}}\hspace{2mm} ln \hspace{2mm} \Big( \frac{\bar{L}\|\epsilon\|_\infty}{\bar{L}(1 + \|\bar{B}\lambda_2^{-1} \beta^{-1}\lambda^{T}\|) \|x(t_k)\| + \|\bar{B}\| \mu} +1 \Big)
\end{equation}
As the right hand side of (\ref{eq:Tk}) is always positive, it is, therefore concluded that inter event execution time is bounded below by a finite positive quantity \cite{gdyn}. This concludes the proof.
\end{proof}

\section{Numerical Simulation}
The efficacy of the proposed control scheme is demonstrated by computer simulation of the given model for two scenarios, i.e., operation of CSTR under no disturbances and time varying disturbances. Following parametric values are used in the experiment.\\
$\mu = 25$, $\beta = 0.3$, $D_a = 0.078$, $\gamma = 20$, $B = 8$ and $x_{2_c} = 0$. $d_1$ and $d_2$ are exogenous disturbances of magnitude $0.026 \hspace{1mm}sin(0.1t)$ and $0.037 \hspace{1mm}sin(0.1t)$ respectively. Moreover, these disturbances are fixed by a positive upper bound, i. e., $|d_1| < |d_2| < |d|_\infty $. Surface coefficient weights $\lambda_1$ and $\lambda_2$ are chosen to be 1 and 2 respectively. $\zeta = \xi = 0.8$, $m_1 = 10^{-4}$, $m_2 = 0.2025$, $\psi = 0.5$ and $\varsigma = 0.97$ are taken to be parameters of the triggering rule. The startup reference trajectory to be tracked is taken as
\begin{equation}
y_{ref} = x_{2_{ss}} (1-k_1 \exp(-k_2 t))
\end{equation}
where $x_{2_{ss}} = 2.6516$ is close to an equilibrium point of the system. $k_2$ and $k_2$ are parameters that are dependent on practical restrictions on the reactor. Here $k_1 = 1 = k_2$. In \cite{acs}, it has been shown that the system has multiple steady state equilibrium points, one of which is $(x_1,x_2)=(0.4472,2.7517)$. Hence, in our experiment, we have provided $x_{2_{ss}}$ very close to this equilibrium point. Therefore, $x_1$ must remain close to $0.4472$ to ensure that the proposed control is worthy and robust.
\subsection{System operating in absence of disturbances}
When external disturbances $d_1$ and $d_2$ are zero, the response of the system under control protocol (\ref{eq:controlLawEvent}) has been shown in figures \ref{fig:compositionNoDisturbance} and \ref{fig:temperatureNoDisturbance}. The response is smooth and state trajectories remain close to the desired operating point for all time. It is also evident that the response is fast. The non-uniform sampling instants when the controller gets updated and the sampling interval (inter-event execution time) have been depicted in figure \ref{fig:instantAndIntervalNoDisturbance}.
\begin{figure}[H]
\begin{minipage}{.5\textwidth}
\centering
\includegraphics[width=\textwidth]{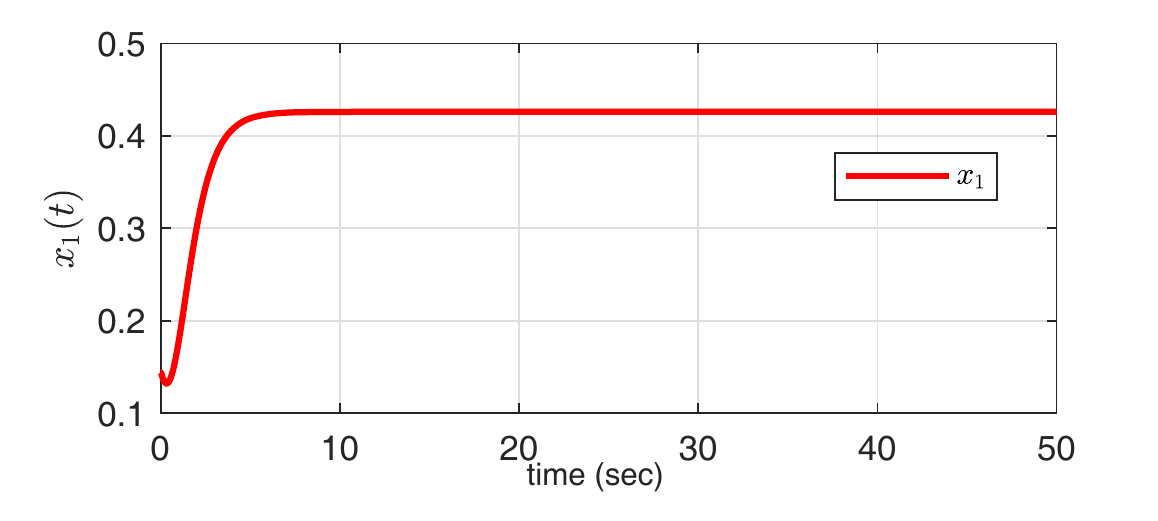}
\caption{Composition profile in absence of disturbances}
\label{fig:compositionNoDisturbance}
\end{minipage}%
\begin{minipage}{.5\textwidth}
\centering
\includegraphics[width=\textwidth]{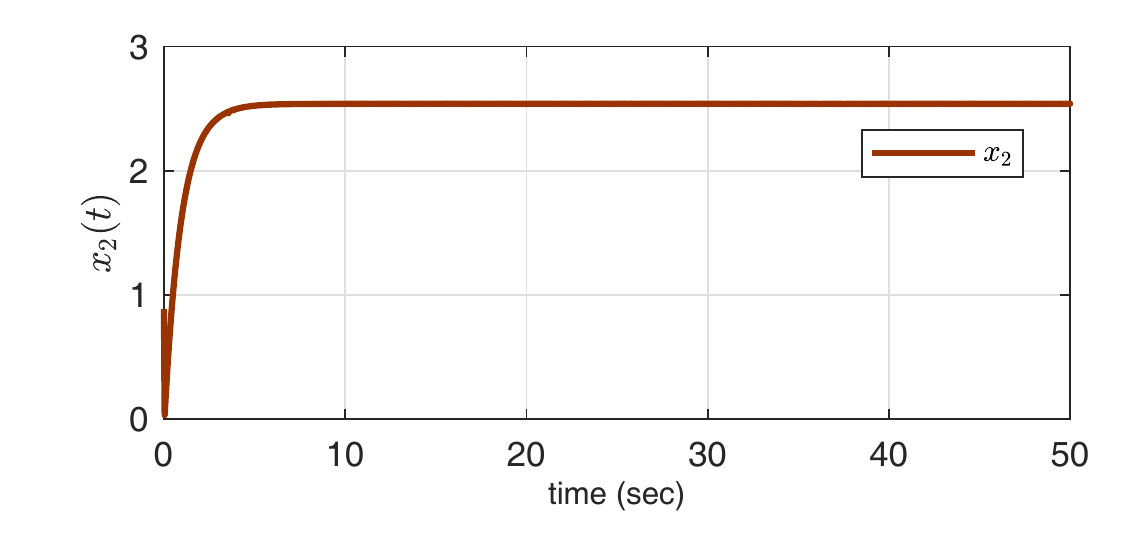}
\caption{Temperature profile in absence of disturbances}
\label{fig:temperatureNoDisturbance}
\end{minipage}%
\end{figure}
\begin{figure}[H]
\centering
\begin{minipage}{.5\textwidth}
\centering
\includegraphics[width=\textwidth]{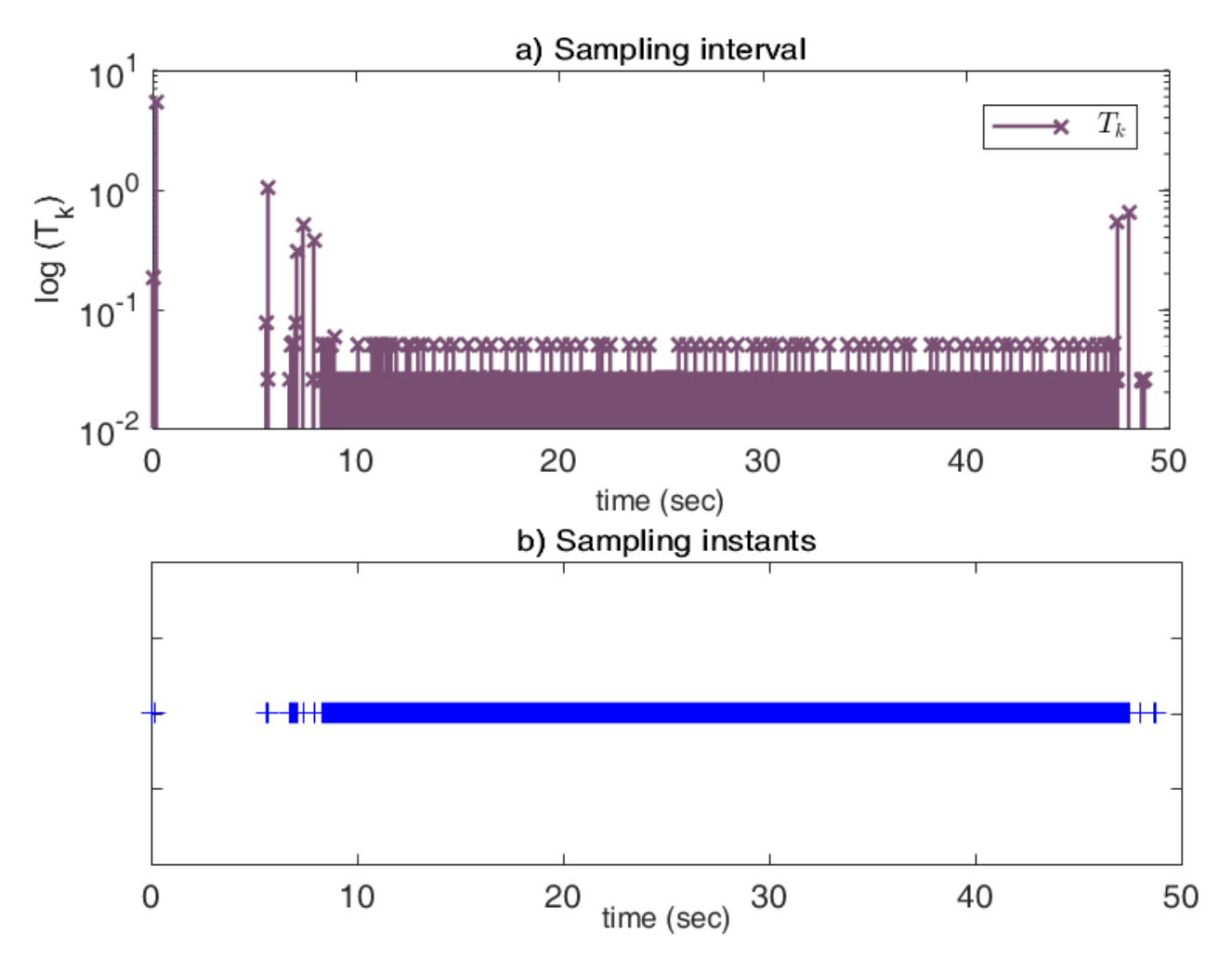}
\caption{Sampling instants and sampling intervals under absence of disturbances}
\label{fig:instantAndIntervalNoDisturbance}
\end{minipage}%
\end{figure}
\subsection{System operating in presence of disturbances}
Figures \ref{fig:compositionDisturbance} and \ref{fig:temperatureDisturbance} show the states of the system under the influence of control signal (\ref{eq:controlLawEvent}) when time varying disturbances $d_1$ and $d_2$ are also taken into account. It is clearly illustrated that the temperature control has been achieved equally well. It has been established earlier in the discussion that control should be synthesized in such a way that after achieving accurate temperature control, the composition should not deviate too much from the desired operating point. Although the desired operating point has been selected as $0.4472$ for the composition state, it can be seen from figure \ref{fig:compositionDisturbance} that the response remains bounded within a band between $0.4067$ and $0.4454$, which is indeed very close to the desired value. Sampling instants and sampling intervals for the case of system operating in the presence of disturbance have been shown in figure \ref{fig:instantAndIntervalDisturbance}. More samples are taken when exogenous disturbance affect the system.
\begin{figure}[H]
\begin{minipage}{.5\textwidth}
\centering
\includegraphics[width=\textwidth]{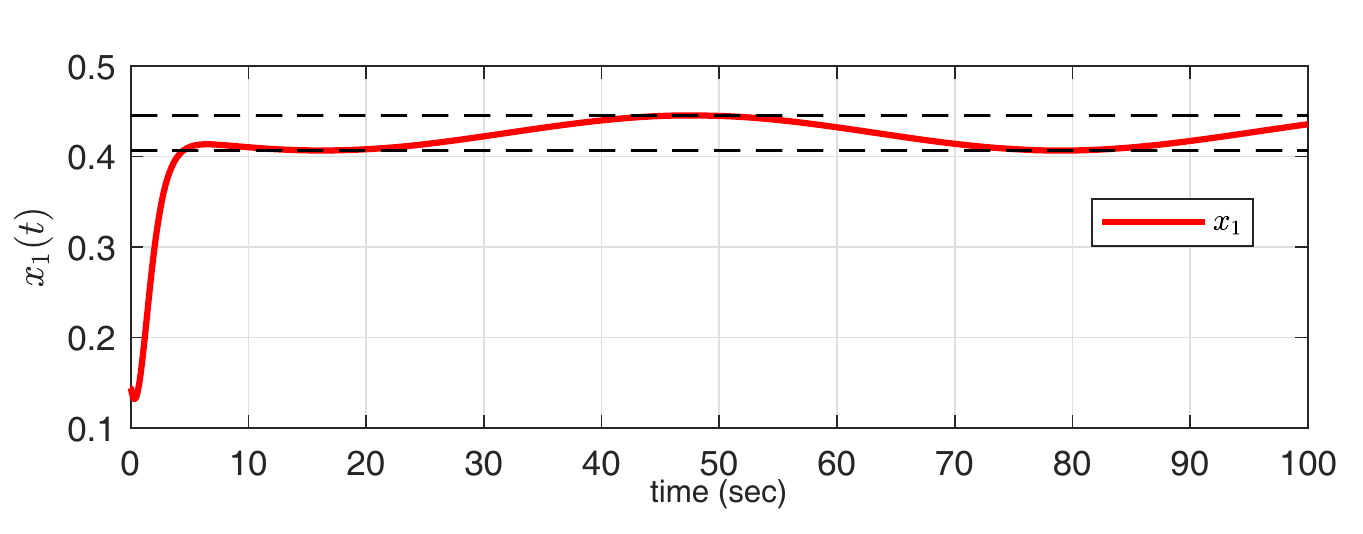}
\caption{Composition profile in presence of disturbances}
\label{fig:compositionDisturbance}
\end{minipage}%
\begin{minipage}{.5\textwidth}
\centering
\includegraphics[width=\textwidth]{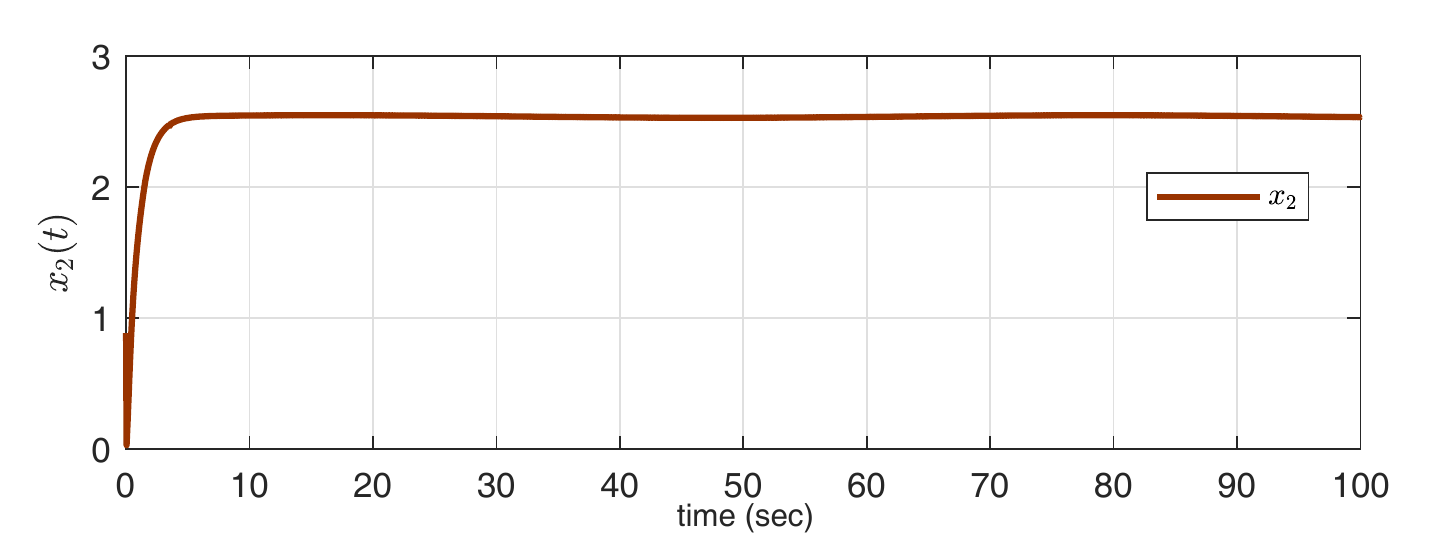}
\caption{Temperature profile in presence of disturbances}
\label{fig:temperatureDisturbance}
\end{minipage}%
\end{figure}
\begin{figure}[H]
\centering
\begin{minipage}{.5\textwidth}
\centering
\includegraphics[width=\textwidth]{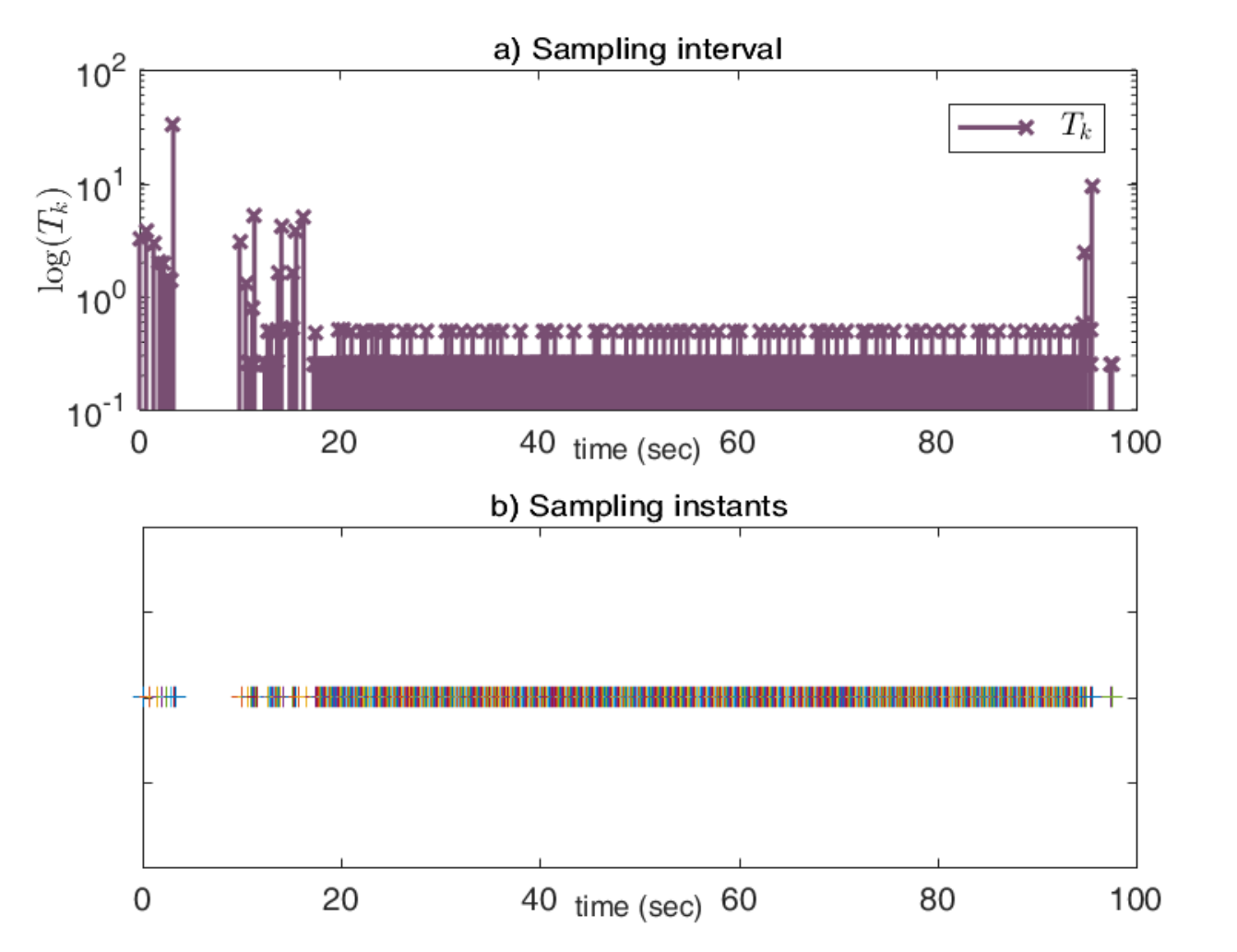}
\caption{Sampling instants and sampling intervals under presence of disturbances}
\label{fig:instantAndIntervalDisturbance}
\end{minipage}%
\end{figure}
Response of the system under the effect of time varying disturbances have been presented in the preceding subsection, wherein it has been clearly confirmed that the control is robust and external disturbances have negligible adverse effect on the desired performance. Further, three cases of temperature regulation have been considered in this work. The temperature was allowed to rise from zero to a set value and the control law (\ref{eq:controlLawEvent}) was used to regulate the temperature at the desired level. Here, we present temperature regulation at $300K$, $400K$ and $500K$ to support our proposition. The regulation cases have been depicted in figures \ref{fig:T300K}, \ref{fig:T400K} and \ref{fig:T500K}. The corresponding sampling instants and sampling intervals have been shown in figures \ref{fig:Tk300K}, \ref{fig:Tk400K} and \ref{fig:Tk500K} respectively. Contrary to periodic update of the controller in time-triggered case, the control signal has been updated with a fresh value only when an event occurs. This significantly reduces the computational and energy expenses, and promotes control by exception. It is observed that the set regulation is achieved quite fast and accurate, as desired.
\begin{figure}[H]
\begin{minipage}{.5\textwidth}
\centering
\includegraphics[width=\textwidth]{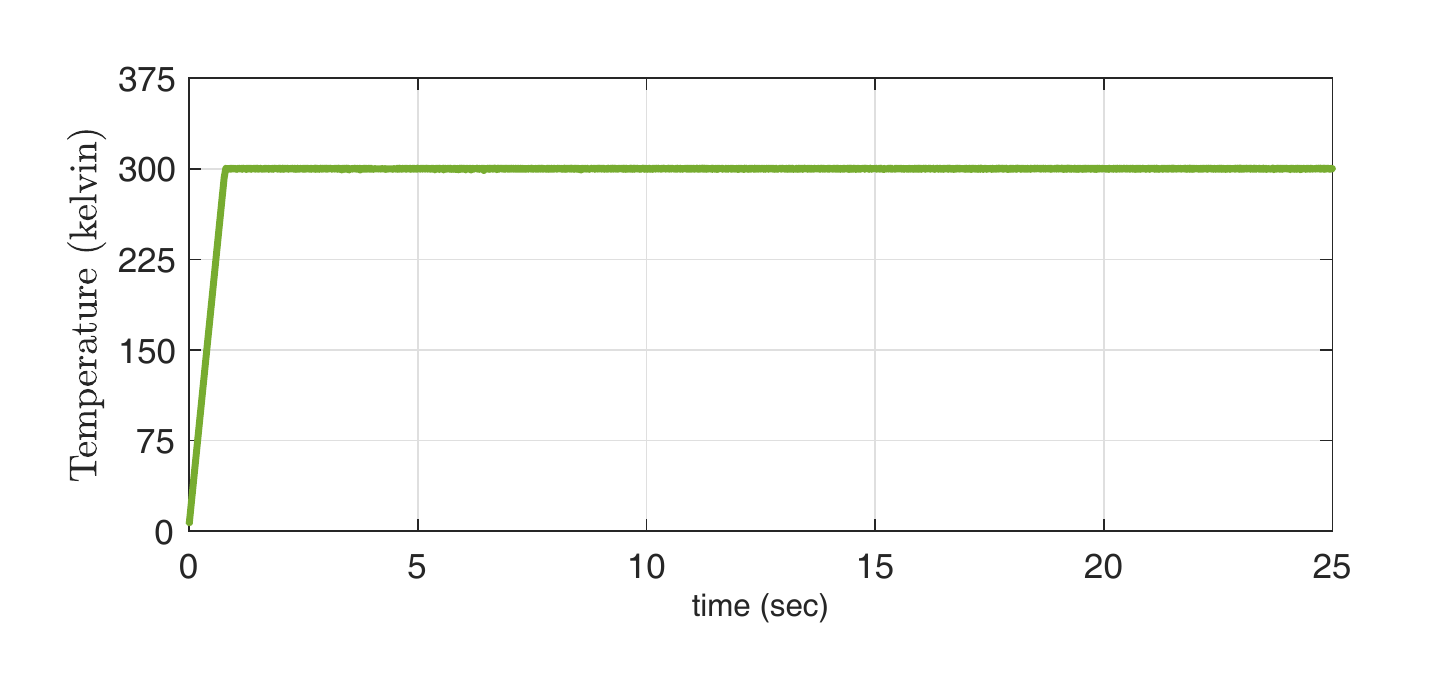}
\caption{Temperature regulation at 300K}
\label{fig:T300K}
\end{minipage}%
\begin{minipage}{.5\textwidth}
\centering
\includegraphics[width=\textwidth]{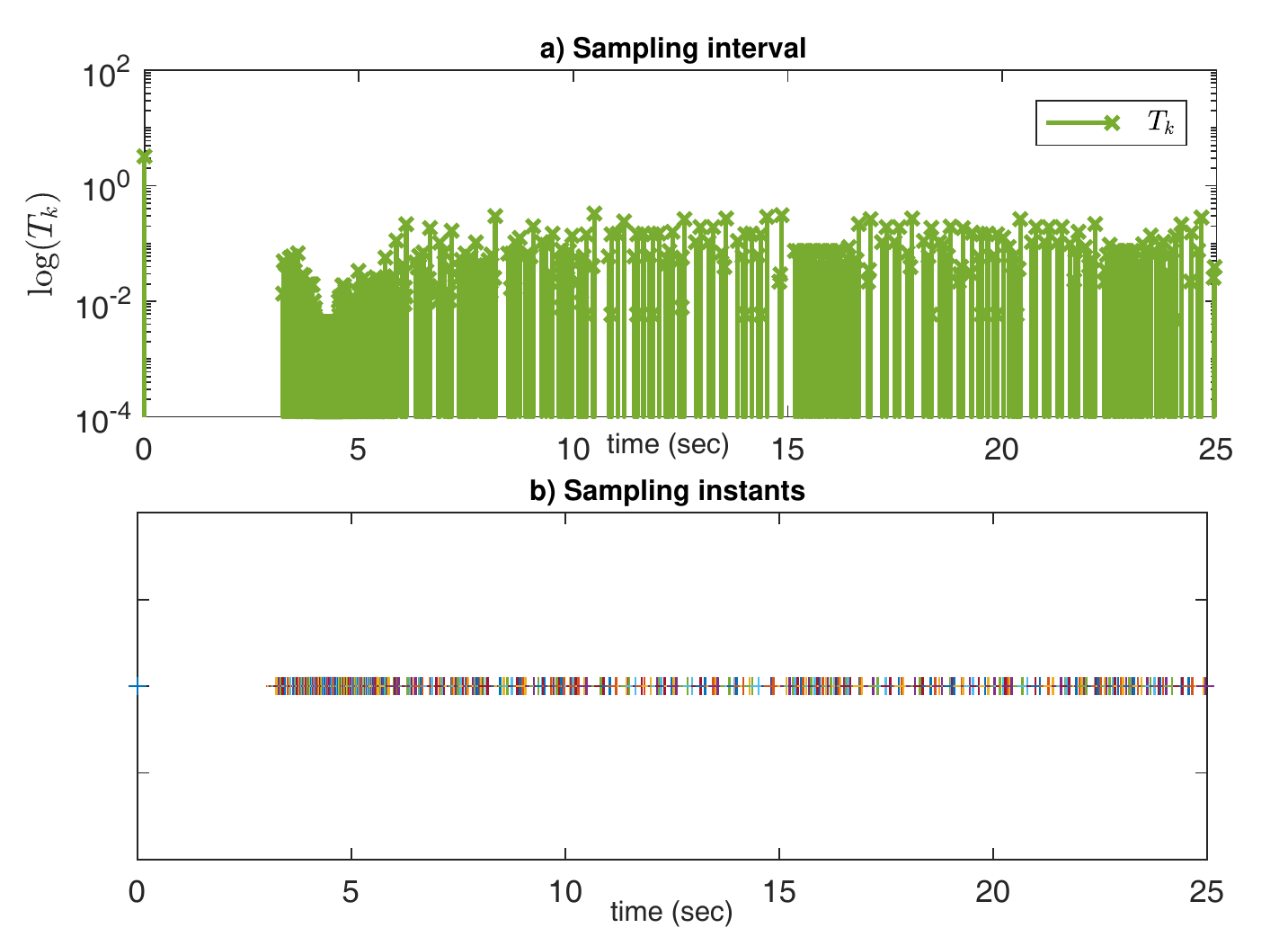}
\caption{Sampling instants and sampling intervals for temperature regulation at 300K}
\label{fig:Tk300K}
\end{minipage}%
\end{figure}
\begin{figure}[H]
\begin{minipage}{.5\textwidth}
\centering
\includegraphics[width=\textwidth]{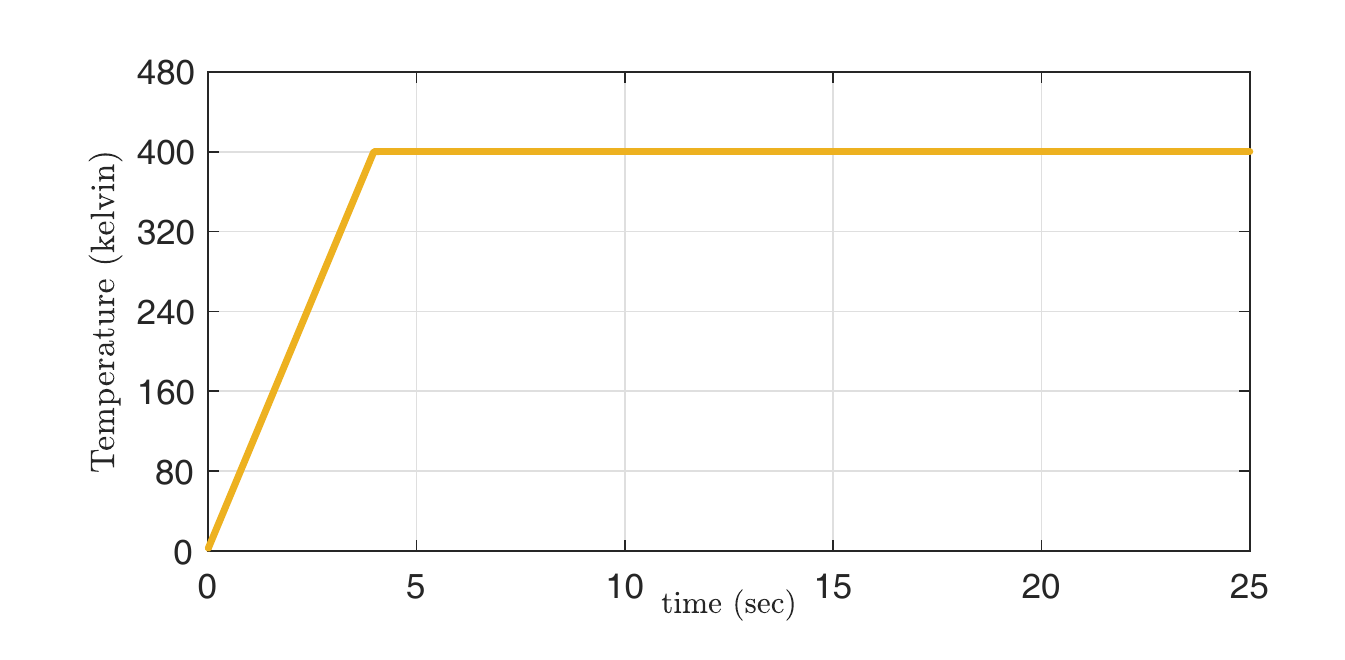}
\caption{Temperature regulation at 400K}
\label{fig:T400K}
\end{minipage}%
\begin{minipage}{.5\textwidth}
\centering
\includegraphics[width=\textwidth]{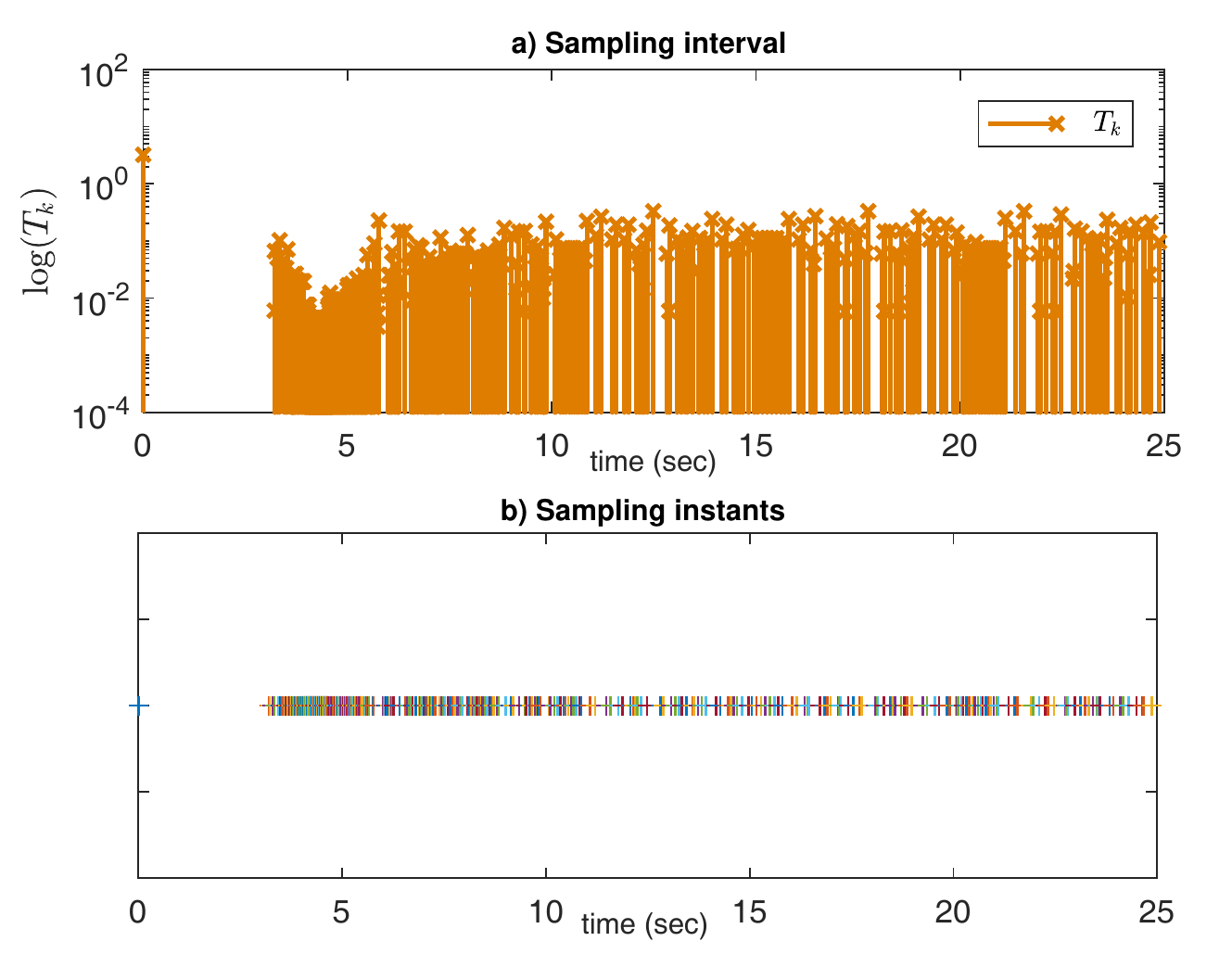}
\caption{Sampling instants and sampling intervals for temperature regulation at 400K}
\label{fig:Tk400K}
\end{minipage}%
\end{figure}
\begin{figure}[H]
\begin{minipage}{.5\textwidth}
\centering
\includegraphics[width=\textwidth]{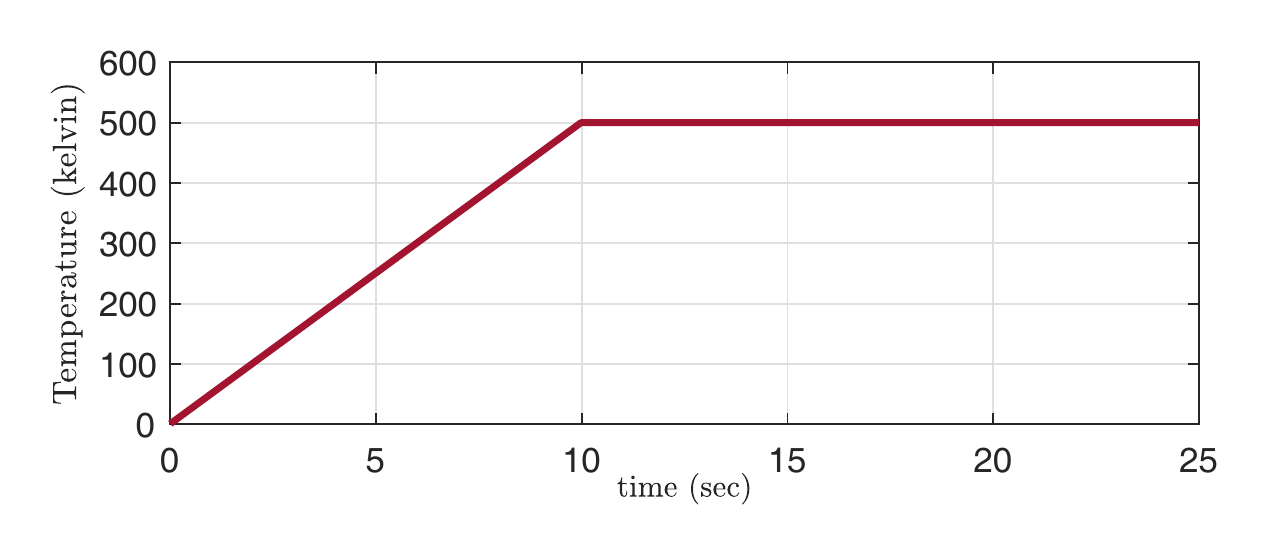}
\caption{Temperature regulation at 500K}
\label{fig:T500K}
\end{minipage}%
\begin{minipage}{.5\textwidth}
\centering
\includegraphics[width=\textwidth]{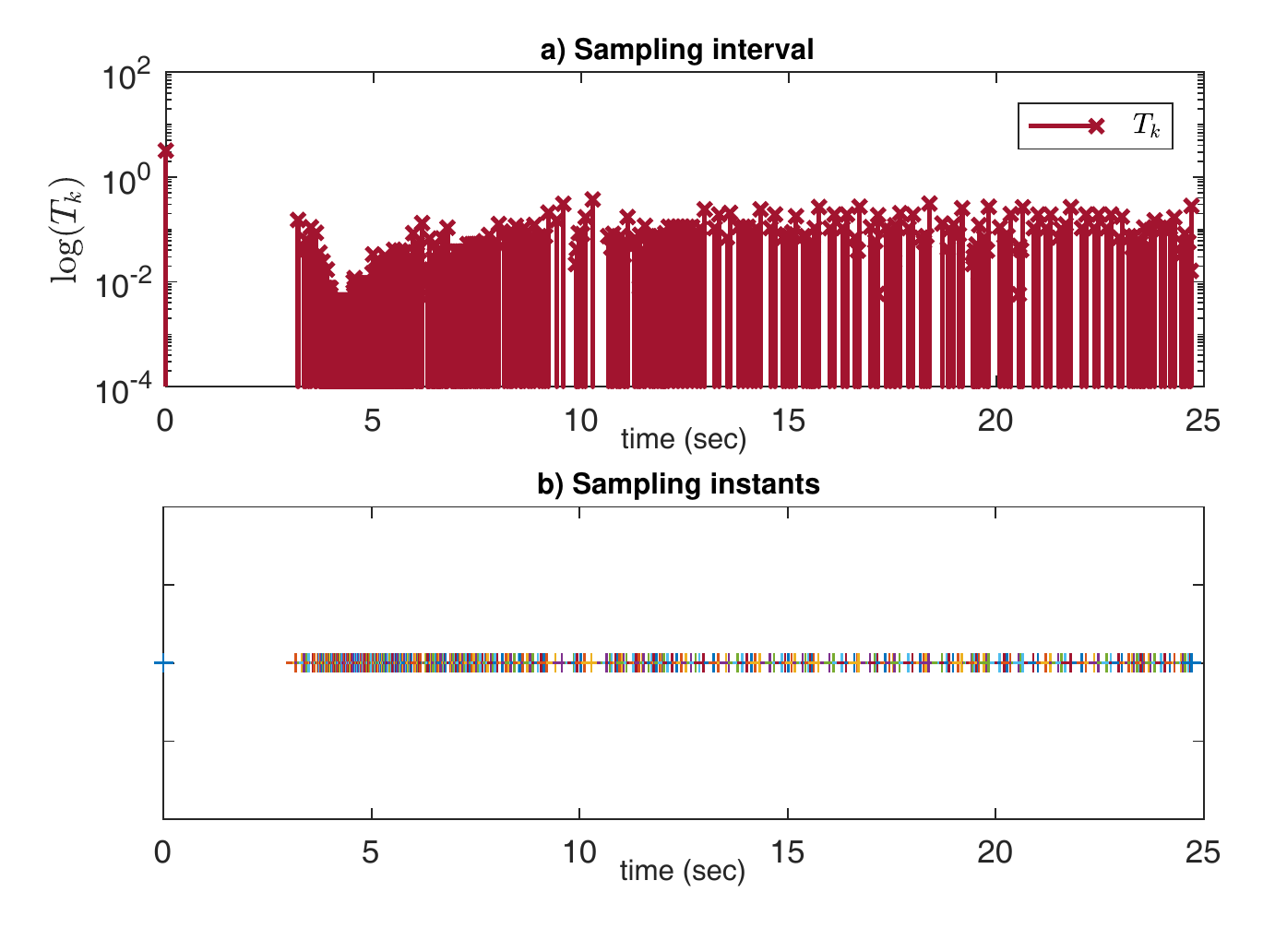}
\caption{Sampling instants and sampling intervals for temperature regulation at 500K}
\label{fig:Tk500K}
\end{minipage}%
\end{figure}
\section{Conclusion}
A novel nonlinear controller based on archetype of event triggered sliding mode has been designed to control a continuous stirred tank reactor. It has been shown that the controller is sturdy and provides stability to the system in a very short span of time. State trajectories have been maintained in close proximity of equilibrium points with minimum computation by the controller. Event triggering technique is one practical control application wherein resource utilization is minimal but optimal closed loop performance is not compromised. The proposed controller based on event triggering SMC provides stability to the system in the sense of Lyapunov. The inter event time is separated by a finite discrete time interval to ascertain no Zeno behavior results.
Numerical simulations are presented to confirm the effectiveness of the proposed event driven sliding mode control.

%
%

\section*{References}
\bibliography{mybibfile}

\begin{thebibliography}{10}
\bibitem{intro}
A quick introduction to sliding mode control and its applications, Tech. rep.,
  Universita' Degli Studi Di Cagliari.
\end{thebibliography}

\end{document}